\documentclass[draftclsnofoot,12pt,onecolumn]{IEEEtran} 
\usepackage{amssymb,amsmath}
\usepackage{cite}
\usepackage{graphicx,subfigure}
\usepackage{psfrag}
\usepackage{url}
\usepackage[latin1]{inputenc}
\usepackage[absolute,overlay]{textpos}
\usepackage{tikz}
\usetikzlibrary{arrows,calc,decorations.markings}
\usetikzlibrary{topaths}
\usetikzlibrary{shapes,trees,snakes}
\usetikzlibrary{arrows}
\usetikzlibrary{shadows}
\usetikzlibrary{positioning}
\usetikzlibrary{matrix}
\usetikzlibrary{shapes.geometric}
\usetikzlibrary{decorations.pathmorphing}
\usepgflibrary{patterns}
\usetikzlibrary{calc}
\usetikzlibrary{fit}					
\usetikzlibrary{backgrounds}	

\usepackage{pgf}
\usetikzlibrary{arrows,automata}
\usepackage[latin1]{inputenc}

\usepackage{amsthm}

\newtheorem{theorem}{Theorem}

\begin{document}

\title{Ultra-Reliable Short-Packet Communications with Wireless Energy Transfer}
\author{
	\IEEEauthorblockN{	Onel L. Alcaraz López, %
						Hirley Alves, 
						Richard Demo Souza, and
						Evelio Martín García Fernández
					}
					%
	\thanks{O.L.A. López, E.M.G. Fernández are with Federal University of Paraná (UFPR), Curitiba, Brazil. \{onel.luis,evelio@ufpr.br\}.} 
	\thanks{ R.D. Souza is with Federal University of Technology - Paraná (UTFPR), Brazil. \{richard@utfpr.edu.br\}.}
	\thanks{H. Alves is with the Centre for Wireless Communications (CWC), University of Oulu, Finland. \{hirley.alves@oulu.fi\}.}
	\thanks{This work was supported by CNPq, CAPES, Funda\c{c}\~ao Arauc\'aria (Brazil) and Academy of Finland, and the Program for Graduate Students from Cooperation Agreements (PEC-PG, of CAPES/CNPq Brazil).}
}					
					
\maketitle

\begin{abstract}
We analyze and optimize a wireless system with energy transfer in the downlink and information transfer in the uplink, under quasi-static Nakagami-m fading. We consider ultra-reliable communication scenarios representative of the fifth-generation of wireless systems, with strict error and latency requirements. The error probability and delay are investigated, and an approximation for the former is given and validated through simulations. The numerical results demonstrate that there are optimum numbers of channels uses for both energy and information transfer for a given message length.
\end{abstract}

\section{Introduction}
Wireless energy transfer (WET) is emerging as a potential technology for powering small and energy-efficient devices with low power requirements through radio frequency (RF) signals~\cite{Ulukus.2015}. This comes with the implicit advantage that RF signals can carry both energy and information~\cite{Varshney.2008}. This scenario becomes very attractive for future communication paradigms such as the Internet of Things, where powering a potentially massive number of devices will be a major challenge \cite{Zanella.2014}. 
 
As pointed out in \cite{Makki.2016}, the most important characteristics of WET systems are: \textit{i}) power consumption of the nodes is in the order of $\mu$W; \textit{ii}) strict requirements on the reliability of the energy supply and of the data transfer; \textit{iii}) information is conveyed in short packets. This third requirement is due to intrinsically small data payloads, low-latency requirements, and/or lack of energy resources to support longer transmissions \cite{Khan.2016}. Indeed, short packets are essential to support Ultra-Reliable Communication (URC) \cite{Popovski.2014}, which is a novel operation mode under discussion for the fifth-generation (5G) of wireless systems. URC over a Short Term (URC-S) focuses on how to deliver short packets under stringent latency ($\le 10$ms) and error probability (e.g., $10^{-5}$) requirements~\cite{Durisi.2015}.

Although performance metrics like Shannon capacity, and its extension to nonergodic channels, have been proven useful to design current wireless systems, they are no longer appropriate in a short-packet scenario \cite{Durisi.2015}. This is because such performance metrics are asymptotic in the packet length and fit well for systems without hard delay requirements~\cite{Devassy.2014}. The Shannon capacity implies that an arbitrarily low error probability can be achieved when sufficiently long packets are used, i.e., introducing sufficiently long delays. In the case of short packets, a more suitable metric is the maximum achievable rate  at a given block length and error probability. This metric has been characterized in \cite{Polyanskiy.2010,Yang.2014} for both Additive White Gaussian Noise (AWGN) and fading channels. In addition, the authors of \cite{Schiessl.2016,Hu.2016_2} incorporate the cost of acquiring instantaneous channel state information (CSI) within a transmission deadline, and analyze the impact of a target error probability for different scenarios under a finite blocklength regime.

WET systems with short packets have been recently investigated in the literature. In \cite{Tandon.2016} subblock energy-constrained codes are investigated, and a sufficient condition on the subblock length to avoid energy outage at the receiver is provided. In \cite{Khan.2016}, a node, charged by a power beacon, attempts to communicate with a receiver over a noisy channel. The system performance is investigated as a function of the number of channel uses for WET and for wireless information transfer (WIT).
An amplify-and-forward relaying scenario is analyzed in \cite{Haghifam.2016} and tight approximations for the outage probability/throughput are given.
Retransmission protocols, in both energy and information transmission phases, are implemented in \cite{Makki.2016} to reduce the outage probability compared to open-loop communications. Optimal power allocation and time sharing between the energy and information signals is proposed to minimize the energy-constrained outage probability. 

Differently from the above works, this paper aims at energy constrained URC-S scenarios with both error probability and latency constraints. 
The system is composed of a point-to-point communication link under Nakagami-m quasi-static fading, with WET in the downlink and WIT in the uplink. The main contributions are: 1) we derive a closed-form approximation for the WIT error probability, as a function of the amount of channel uses in the WET and WIT phases, and validate its accuracy through simulations; 2) we show  the existence of  optimum values for the amount of channel uses in the WIT and WET phases, and that by increasing the WIT blocklength the required WET blocklength for a given target reliability decreases. Our results also show that the more stringent the reliability requirement, the higher the required delay, increasing also the portion of time dedicated to the WET phase. Moreover, in general the possibility of meeting the reliability and latency constraints increases by decreasing the message length. 

Next, Section \ref{system} presents the system model and assumptions. Section \ref{FB} discusses the performance metrics and an approximation for the error probability. Section \ref{results} presents the numerical results. Finally, Section \ref{conclusions} concludes the paper.
\newline\textbf{Notation:} $X\sim\Gamma(m,1/m)$ is a normalized gamma distributed random variable with shape factor $m$ and Probability Density Function (PDF) $f_X(x)=\frac{m^m}{\Gamma(m)}x^{m-1}e^{-mx}$. Let $\mathbb{E}[\mathrel{\cdot}]$ denote expectation, while $\ _xF_y$ and $\operatorname{K}_t(\mathrel{\cdot})$ are the generalized hypergeometric function and the modified Bessel function of second kind and order $t$, respectively \cite{Jeffrey.2007}.
\section{System Model}\label{system}
Consider the scenario in Fig.~\ref{Fig1}, in which $S$ represents the information source while $D$ is the destination. The nodes are single antenna, half-duplex, devices. $D$ is assumed to be externally powered and acts as an interrogator, requesting information from $S$, which may be seen as a sensor node with very limited energy supply. First, $D$ charges $S$ during $v$ channel uses in the WET phase, through channel $h$. Then, $S$ uses the energy obtained in the WET phase to transmit $k$ information bits over $n$ channel uses in the WIT phase\footnote{Different (and possibly very distant) frequency bands for both processes are assumed, motivated by the different properties of the WET and WIT reception/transmission circuits as in \cite{Makki.2016,Chen.2015,Mandal.2008,Moritz.2014}.}. through channel $g$. The duration of a channel use is denoted by $T_c$.
Nakagami-m quasi-static channels are assumed, where the fading process is considered to be constant over the transmission of a block (either of $v$ or $n$ channel uses) and independent and identically distributed from block to block. We consider independent normalized channel gains, then $h^2\sim\Gamma(m,1/m)$ and $g^2\sim\Gamma(m,1/m)$. In addition, perfect CSI at $D$ is assumed in the decoding after the WIT phase\footnote{CSI acquisition in an energy-limited setup is not trivial and including the effect of imperfect CSI would demand a more elaborated mathematical analysis that is out of the scope of this work. However, notice that when channels remain constant over multiple transmission rounds, the cost of CSI acquisition can be negligible. Anyway, our analysis based on perfect CSI gives an upper-bound on the performance of real scenarios.}.
\begin{figure}[t!]
	\centering
	\subfigure{\includegraphics[width=0.8\textwidth]{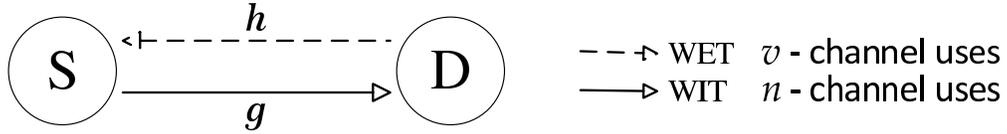}}
	\vspace*{-2mm}
	\caption{System model with WET in the downlink and WIT in the uplink.}		
	\label{Fig1}
	\vspace*{-2mm}
\end{figure}
\subsection{WET Phase}
In this phase, $D$ charges $S$ during $v$ channel uses. The energy harvested at $S$ is given by
\begin{equation}\label{Eh}
E=\frac{\eta P_{_D} h^2}{\kappa d^{\alpha}}vT_c,
\end{equation}
where $P_{_D}$ is the transmit power of $D$, $0<\eta<1$ is the energy conversion efficiency, $d$ is the distance between $S$ and $D$, $\alpha$ is the path loss exponent and $\kappa$ accounts for other factors as the carrier frequency, heights and gains of the antennas \cite{Goldsmith.2005}. In addition, we assume that $P_{_D}$ is sufficiently large such that the energy harvested from noise is negligible. 
\subsection{WIT Phase}
After the WET phase, $S$ uses the harvested energy to transmit a message of $k$ bits to $D$ over $n$ channel uses. The signal received at $D$ can be written as
\begin{equation}\label{yd}
y_{_D}=\sqrt{\frac{P_{_S}}{\kappa d^{\alpha}}}gx_{_S}+w_{_D},
\end{equation}
where $x_{_S}$ is the zero-mean, unit-variance Gaussian codebook transmitted by $S$, $\mathbb{E}[|x_{_S}|^2]=1$, $w_{_D}$ is the Gaussian noise vector at $D$ with variance $\sigma_{_D}^2$ and the transmit power is
\begin{equation}\label{Ps}
P_{_S}=\frac{E}{nT_c}=\frac{\eta vP_{_D}h^2}{n\kappa d^{\alpha}}.
\end{equation}
Thus, the instantaneous Signal-to-Noise Ratio (SNR) at $D$ is
\begin{equation}\label{SNR}
\gamma=\frac{\eta vP_{_D}h^2g^2}{n\kappa^2d^{2\alpha}\sigma_{_D}^2}=\frac{\eta vP_{_D}\tilde{h}\tilde{g}}{m^2n\kappa^2d^{2\alpha}\sigma_{_D}^2}=\mu z,
\end{equation}
where $\tilde{u}=m\cdot u^2$, $u\in\{h,g\}$, is distributed according to the standard gamma PDF, $\mu=\frac{\eta vP_{_D}}{m^2n\kappa^2d^{2\alpha}\sigma_{_D}^2}$ and $z=\tilde{h}\tilde{g}$, whose PDF is given
by \cite[Th. 2.1]{Withers.2013} 
\begin{align} 
	\label{fz}
	f_Z(z) &= \frac{2}{\Gamma(m)^2}z^{m-1} \operatorname{K}_0(2\sqrt{z}),\ z>0.
\end{align}
\section{Performance at Finite Blocklength}\label{FB}
We consider a time-constrained setup, which implies that $D$ has to decode the received signal block by block. The system performance is characterized in terms of error probability and delay when $S$ is transmitting at a fixed rate, $r = k/n$.
\subsection{Error Probability and Delay}
Let $\epsilon$ be the average error probability which, for quasi-static fading channels and sufficiently large values of n, e.g $n\ge100$, can be approximated as \cite[eq.(59)]{Yang.2014}
\begin{align}
\epsilon\approx\mathbb{E}\!\Biggl[Q\!\Biggl(\!\frac{C(\gamma)\!-\!r}{\sqrt{V(\gamma)/n}}\!\Biggl)\!\Biggl]\!=\!\int\limits_{0}^{\infty}\!Q\Biggl(\!\frac{C(\mu z)\!-\!r}{\sqrt{V(\mu z)/n}}\!\Biggl)\!f_Z(z)\mathrm{d}z, \label{EX}
\end{align}
where $C(\gamma)=\log_2(1+\gamma)$ is the Shannon capacity, $V(\gamma)=\left(1-\frac{1}{(1+\gamma)^2}\right)(\log_2e)^2$ is the channel dispersion, which measures the stochastic variability of the channel relative to a deterministic channel with the same capacity \cite{Polyanskiy.2010}, and $Q(x)=\int_{x}^{\infty}\frac{1}{\sqrt{2\pi}}e^{-t^2/2}\mathrm{d}t$. Let $\delta$ be the delay in delivering a message of $k$ bits, while $\delta^*$ is the minimum delay that satisfies a given reliability constraint. Moreover, $\nu$ is the time sharing parameter representing the fraction of $\delta$ devoted to WET only. Both metrics are given next 
\begin{align}
\delta&=n+v, \label{delta}\\
\nu&=v/\delta. \label{nu}
\end{align}
Notice that $\delta$ is measured in channel uses, while $\delta T_c$ would be the delay in seconds. Finally, we define the optimum WIT blocklength, in the sense of minimizing $\delta^*$, as $n^*$. Both $\delta^*$ and $n^*$ are numerically investigated in Section \ref{results}.
\subsection{Error probability approximation}
It seems intractable to find a closed-form solution for \eqref{EX}. Then, first we resort to an approximation of  $Q(p(\mu z))$, $p(\mu z)=\frac{C(\mu z)-r}{\sqrt{V(\mu z)/n}}$, given by \cite{Makki.2014,Makki.2016}
\begin{eqnarray}\label{AP}
Q(p(\mu z))\!\approx\!\Omega(\mu z)\!=\!\left\{\begin{array}{ll}
1,&  z\le \zeta^2\\
\frac{1}{2}\!-\!\frac{\beta}{\sqrt{2\pi}}(\mu z\!-\!\theta)\!,\!&\!\zeta^2\!<\!z\!<\!\xi^2\!\\
0,& z\ge \xi^2
\end{array}
\right.\!, 
\end{eqnarray}
where $\zeta^2 = \tfrac{\varrho}{\mu} $,  $\xi^2 = \tfrac{\vartheta}{\mu}$,  $\theta=2^{k/n}-1$, $\beta=\sqrt{\frac{n}{2\pi}}(2^{2k/n}-1)^{-\frac{1}{2}}$, $\varrho=\theta-\frac{1}{\beta}\sqrt{\frac{\pi}{2}}$ and $\vartheta=\theta+\frac{1}{\beta}\sqrt{\frac{\pi}{2}}$, which leads to the following result.
\begin{theorem}\label{prop_1}
	For the system described in Section \ref{system}, the error probability in \eqref{EX} can be approximated as 
	%
	\begin{align}\label{AP2} 
	\mathbb{E}[\epsilon]&\approx\zeta^{2m}\tfrac{(\omega_1-\omega_3)}{m}\bigg[\operatorname{K}_0(2\zeta)\ _1F_2(1;m,m\!+\!1;\zeta^2)+\tfrac{\zeta}{m}\operatorname{K}_1(2\zeta)\ _1F_2(1;m\!+\!1,m\!+\!1;\zeta^2)\bigg]\! \nonumber \\ &\quad+\!\zeta^{2m\!+\!2}\tfrac{\omega_2}{m\!+\!1}\bigg[\operatorname{K}_0(2\zeta)  _1F_2(1;m\!+\!1,m\!+\!2,\zeta^2)\!+\!\tfrac{\zeta}{m\!+\!1}\operatorname{K}_1(2\zeta)\ _1F_2(1;m\!+\!2,m\!+\!2,\zeta^2)\bigg] \nonumber \\
	&\quad+\!\xi^{2m}\operatorname{K}_0(2\xi)\bigg[\tfrac{\omega_3}{m}\ _1F_2(1;m,m\!+\!1;\xi^2) \!-\!\tfrac{\omega_2\xi^2}{m+1}\!\ \!_1F_2(1;\!m\!+\!1,\!m\!+\!2;\xi^2)\!\bigg] \nonumber \\ &\quad+\!\xi^{2m\!+\!1}\!\operatorname{K}_1(2\xi)\!\bigg[\tfrac{\omega_3}{m^2}\!\ _1F_2(1;\!m\!+\!1,\!m\!+\!1;\xi^2)\!-\!\tfrac{\omega_2\xi^2}{(m\!+\!1)^2}\ \!_1F_2(1;\!m\!+\!2,\!m\!+\!2;\xi^2)\bigg] \tag{10}
	\end{align}
\end{theorem}
\begin{proof}
	See Appendix~\ref{App_A}.
\end{proof}

\section{Numerical Results}\label{results}
In the numerical results we consider scenarios with stringent error probability requirements, which are expected to be typical of URC-S services in future 5G systems. Therefore, being $\epsilon_{_0}$ the target error probability,  $\epsilon\le\epsilon_{_0}$ must be satisfied.
Results are obtained by setting $m=3$, $\alpha=3$ and $d=12$m. We also assume that $\kappa=10^3$, what is equivalent to 30 dB average signal power attenuation at a reference distance of 1 meter. Following the state-of-the-art in circuit design, we consider $\eta=0.5$~\cite{Lu.2015}. Moreover, $P_{_D}=30$dBm and $T_c=3\mu$s, thus $\sigma_{_D}^2=-110$dBm is a valid assumption if a bandwidth around 1MHz is assumed.
\begin{figure}[!t]
	\centering
	\subfigure{\label{Fig2a}\includegraphics[width=0.85\textwidth]{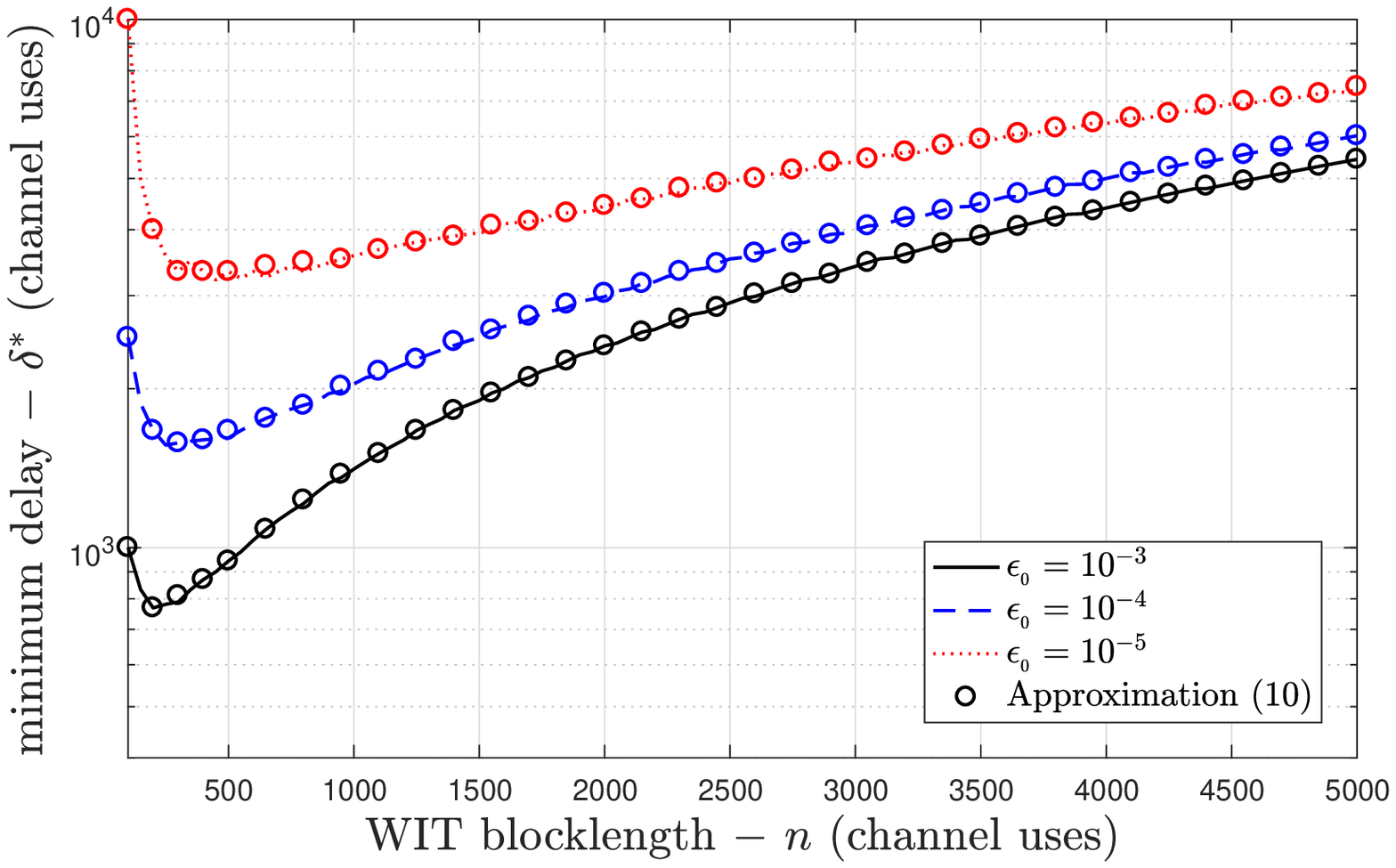}}
	\vspace{-2mm}
	\subfigure{\label{Fig2b}\includegraphics[width=0.85\textwidth]{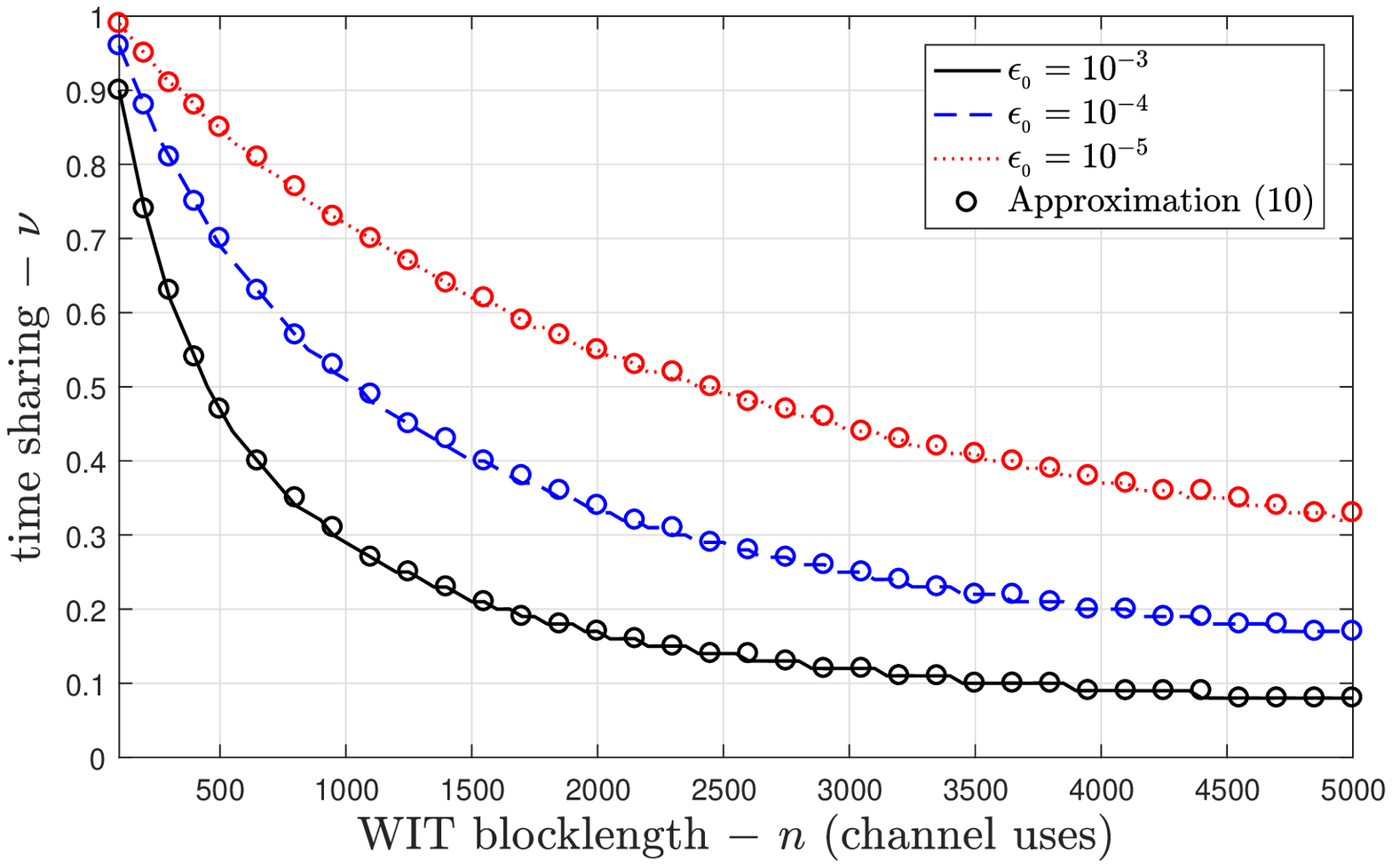}}	
	\caption{(a) $\delta^*$ (top) and (b) $\nu$ (bottom), for packets of $k=216$ bits, as a function of the WIT blocklength $n$.}\label{Fig2}	
\end{figure}
\begin{figure}[!t]
	\centering
	\subfigure{\includegraphics[width=0.85\textwidth]{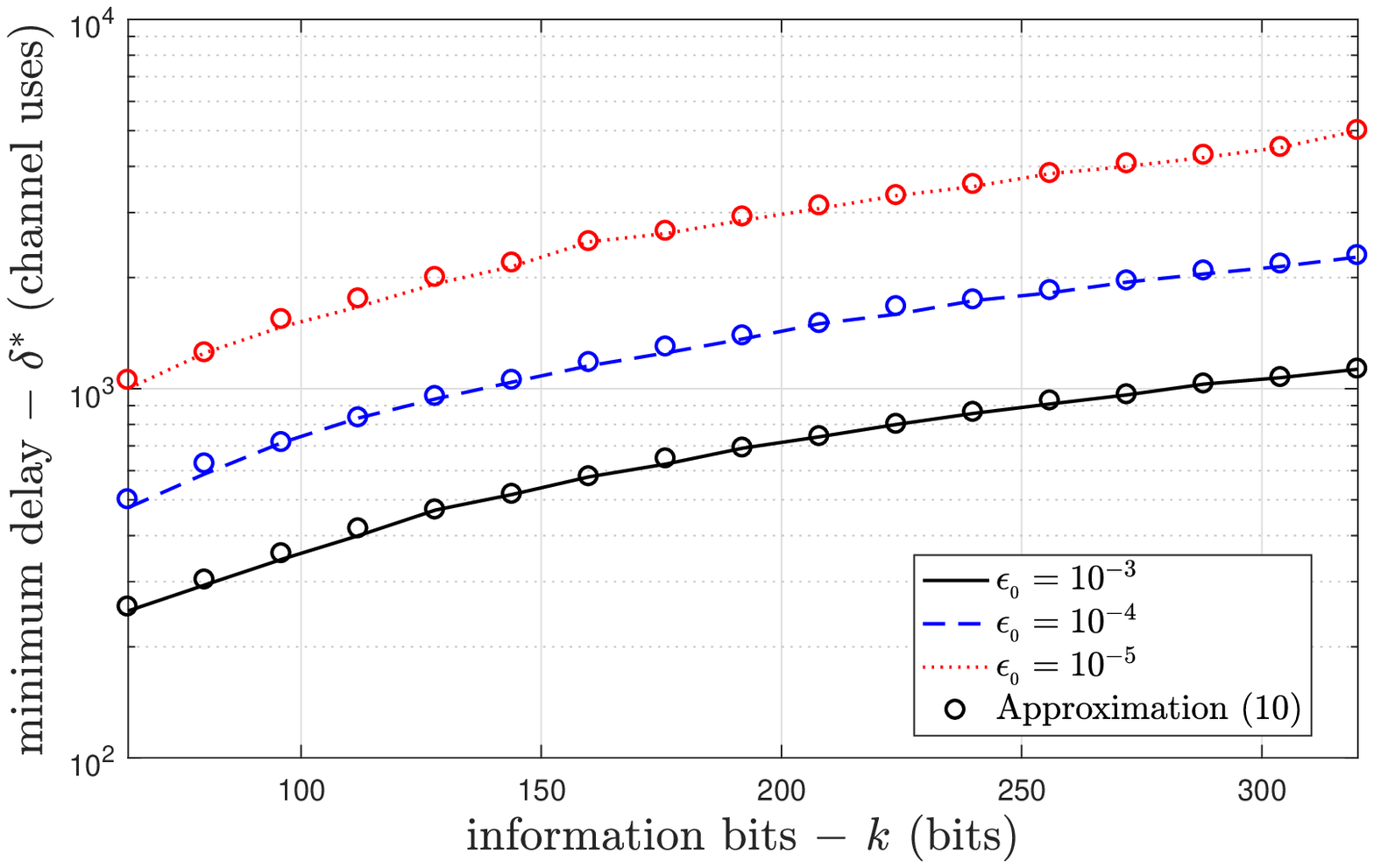}}
	\vspace{-2mm}
	\caption{$\delta^*$ as a function of message length $k$ for different $\epsilon_{_0}$.}		
	\label{Fig3}
\end{figure}

Fig.~\ref{Fig2} shows $\delta^*$, and the corresponding $\nu$, when messages of $k=216$ bits are transmitted with $100\le n\le 5000$ channel uses. Particularly, Fig.~\ref{Fig2}a (top) shows the minimum required delay $\delta^*$, while the corresponding time sharing parameter $\nu$ is plotted in Fig.~\ref{Fig2}b (bottom). We can note that there is an optimum value for the WIT blocklength, $n^*$, which increases as $\epsilon_{_0}$ decreases but always keeps relatively small ($n^*<500$ channel uses) and, at the same time, the optimum WET time increases. As shown in Fig.~\ref{Fig2}b, when $n$ increases, the portion of time devoted for WET decreases and the optimum WET blocklength remains very small for $\epsilon_{_0}=10^{-3}$. When $\epsilon_{_0}=10^{-5}$ (reliability of $99.999\%$), the minimum allowable delay is about $3300$ channel uses $\sim 10$ms, due to the large value of the required WET blocklength, which could be very severe for some URC-S scenarios. Moreover, it is interesting to note that the approximation derived in \eqref{AP2} agrees very well with the numerical integration of \eqref{EX}, which is accurate for $n>100$ for AWGN \cite[Figs.~12 and 13]{Polyanskiy.2010} and fading channels \cite{Yang.2014}.

Retransmission protocols can be very convenient to improve the system performance as investigated in \cite{Makki.2016,Makki.2014} for other scenarios. In that sense, an error probability of $10^{-6}$ can be achieved via $\epsilon_{_0}=10^{-3}$ by allowing for just one packet retransmission and causing a delay much smaller than 3000 channel uses. However, such performance improvement would only be possible if the retransmissions are subject to independent channel realizations (i.e., if the channel coherence time is not larger than the blocklength). In our work we aim at showing the performance in a simple open-loop scenario where the channel coherence time is larger than the blocklength. On the other hand, and as shown in Fig.~\ref{Fig3}, reducing the message length $k$ helps in improving the minimum required delay for a given $\epsilon_{_0}$, which is also in line with previous results for a system without delay restrictions \cite[Fig.~8]{Makki.2016}. 
For instance, if $k=128$ information bits, then $\delta^*\approx 2000$ channel uses $\sim 6$ms, which meets more stringent system requirements. For a given reliability requirement, and based on \eqref{EX}, when $k$ increases $n$ should increase as well in order to diminish the rate $r$, hence $\gamma$ and $C(\gamma)$ tend to decrease. Also, increasing $v$ slows the capacity decrease and the required $n$ does not need to be so large, thus a trade-off between increasing $v$ and $n$ is identified. Nevertheless, increasing $k$ with a fixed $\epsilon_{_0}$ leads to an inevitable increase in the overall number of channel uses $n+v$, which is also clearly inferred from Fig.~\ref{Fig3}.
\begin{figure}[t!]
	\centering
	\subfigure{\label{Fig4a}\includegraphics[width=0.85\textwidth]{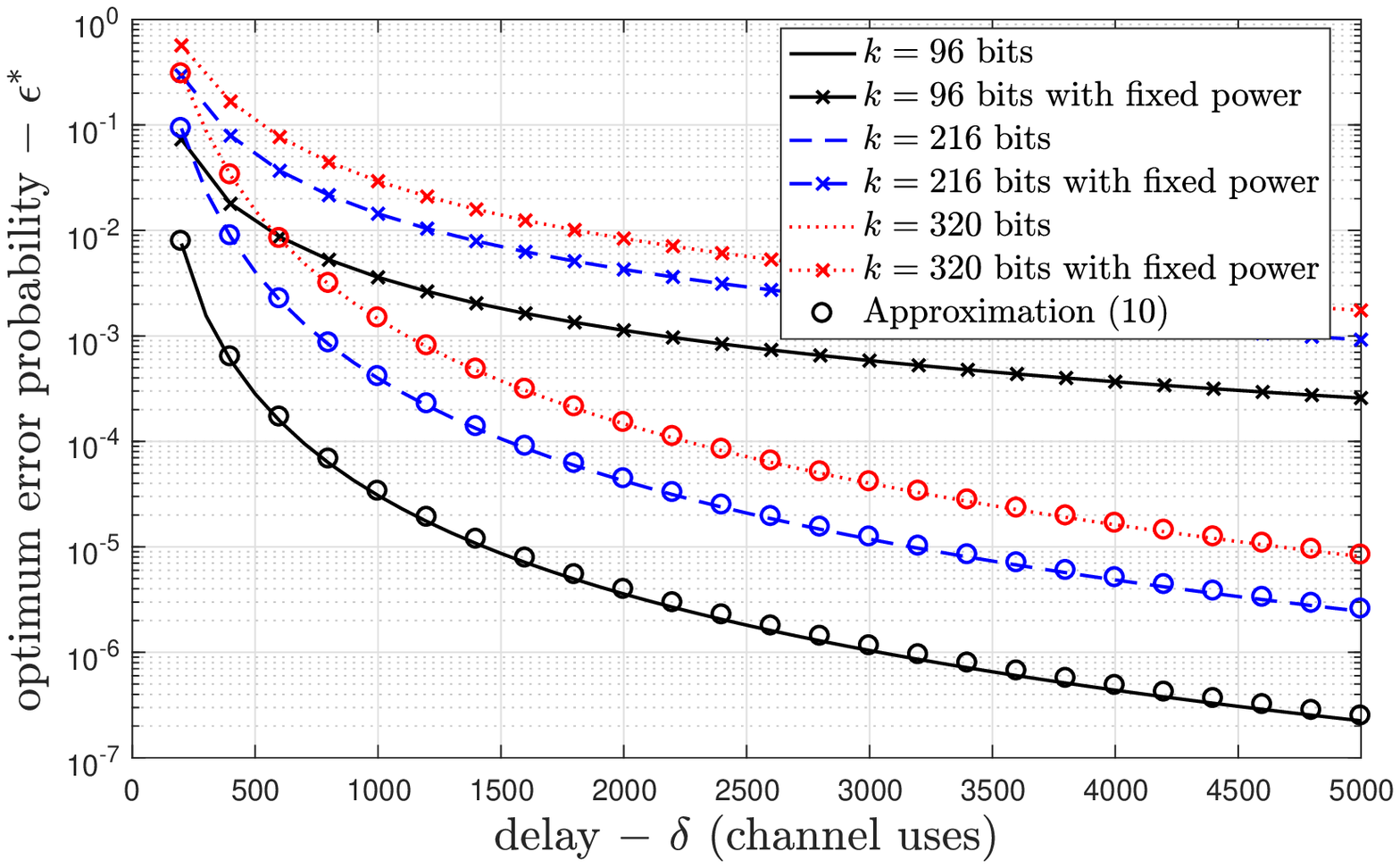}}
	\vspace{-4mm}
	\subfigure{\label{Fig4b}\includegraphics[width=0.85\textwidth]{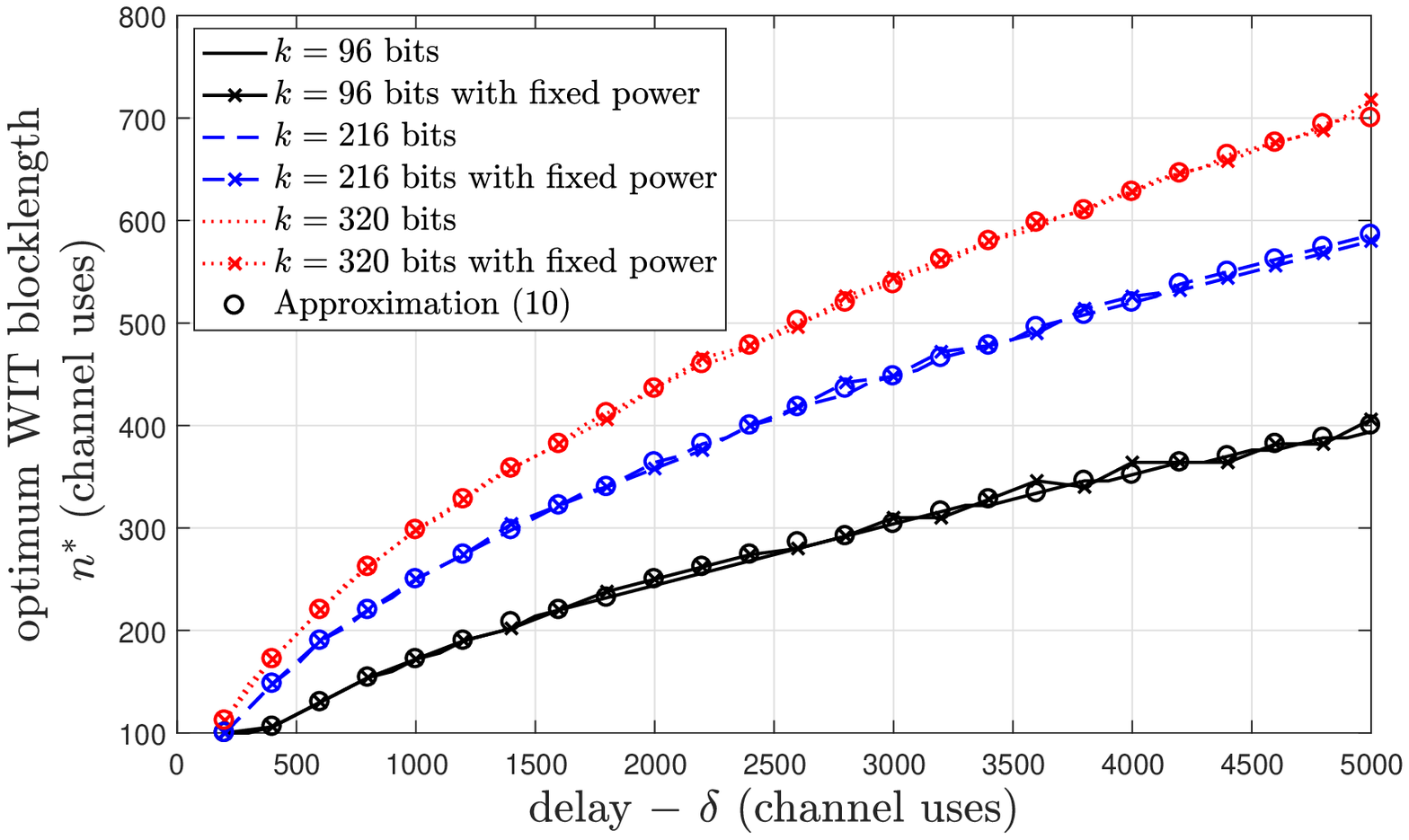}}	
	\caption{(a) $\epsilon^*$ (top) and (b) $n^*$ (bottom), for packets of $k=96$, $k=216$ and $k=320$ bits, as a function of the delay.}\label{Fig4}
	\vspace{-2mm}		
\end{figure}

Finally, in Fig.~\ref{Fig4}a  we evaluate the minimum achievable error probability $\epsilon$ for different delays $\delta$, and we also compare the system performance with that of a scenario where node $S$ transmits with the optimum fixed power $\hat{P}_{_S}^*$. Notice that in such scenario, an outage can also occur when the harvested energy is insufficient for $S$ to transmit with power $\hat{P}_{_S}^*$. In addition, we do not consider energy accumulation from round to round.  As expected, the performance of the fixed transmit power scheme is much worse than that with power allocation.
Fig.~\ref{Fig4}b shows the optimum block length for WIT, $n^*$, which increases slowly with $\delta$, a consequence of the fact that increasing the allowed delay requires increasing more the fraction of time for WET for obtaining a smaller error probability. This means that $v$ has a greater impact on the error probability when increasing the allowable delay than the WIT blocklength $n$ on its own. Also, $n^*$ is practically the same for both scenarios.

The analysis presented in this paper was also carried out using the classical information theoretical tools that assume an infinite blocklength. The differences in the results when considering finite or infinite blocklenghts are small, but increase with the decrease of $k$ or with the increase of $m$, which are inline with the results reported in \cite{Mary.2015} for a scenario without WET, and are omitted here for the sake of brevity. Therefore, we can conclude that the loss in performance due to the use of short packets is very small in our particular scenario.
\section{Conclusion}\label{conclusions}
We evaluated a communication system with WET in the downlink and WIT in the uplink, in URC-S scenarios. We provided an analytical approximation for the error probability and validated its accuracy. The numerical results demonstrated that there are optimum values for the WIT and WET blocklengths. By increasing the WIT blocklength, the required WET blocklength decreases. In addition, we show that the more stringent the reliability, the higher the minimum required delay, which also increases the fraction of time devoted for WET. Finally, we show that depending on the particular requirement, WET in URC-S scenarios may be feasible.
\appendices 
\section{Proof of Theorem~\ref{prop_1}}\label{App_A}
First, let us substitute \eqref{fz} and \eqref{AP} into \eqref{EX}, which yields 
\begin{align}
\mathbb{E}[\epsilon]&\approx \frac{2}{\Gamma(m)^2}\int_{0}^{\infty}z^{m-1}\operatorname{K}_0(2\sqrt{z})\Omega(\mu z)\mathrm{d}z 
\nonumber\\
&\approx\! \omega_1\!\int_{0}^{\zeta^2}\!z^{m\!-\!1}\!\operatorname{K}_0(2\sqrt{z})\mathrm{d}z\!\nonumber\!-\!\omega_2\!\int_{\zeta^2}^{\xi^2}\!z^m\operatorname{K}_0(2\sqrt{z})\mathrm{d}z+\omega_3\int_{\zeta^2}^{\xi^2}z^{m-1}\!\operatorname{K}_0(2\sqrt{z})\mathrm{d}z, \tag{11} \label{AP1}
\end{align}
where $\omega_1=\frac{2}{\Gamma(m)^2}$, $\omega_2=\frac{\beta\mu}{\sqrt{2\pi}}\omega_1$ and $\omega_3=\Big(\frac{1}{2}+\frac{\beta\theta}{\sqrt{2\pi}}\Big)\omega_1$.
In order to continue with the derivations we need first to compute $\int z^p\operatorname{K}_0(2\sqrt{z})\mathrm{d}z$, which is given next. 
\begin{align}
\int z^p\operatorname{K}_0(2&\sqrt{z})\mathrm{d}z\stackrel{(a)}{=}\int\Big(\frac{q}{2}\Big)^{2p+1}\operatorname{K}_0(q)\mathrm{d}q\nonumber\\
&\stackrel{(b)}{=}\Big(\frac{q}{2}\Big)^{2p+2}\Gamma(p\!+\!1)^2\Bigg[\frac{\operatorname{K}_0(q)\ _1F_2\big(1;p\!+\!1,p\!+\!2;\tfrac{q^2}{4}\big)}{\Gamma(p+1)\Gamma(p+2)}+\frac{\frac{q}{2}\operatorname{K}_1(q)\ _1F_2\big(1;p\!+\!2,p\!+\!2;\tfrac{q^2}{4}\big)}{\Gamma(p\!+\!2)^2}\Bigg]\nonumber\\
&\stackrel{(c)}{=}z^{p+1}\frac{\operatorname{K}_0(2\sqrt{z})\ _1F_2(1;p\!+\!1,p\!+\!2;z)}{p+1}+z^{p+\tfrac{3}{2}}\frac{\operatorname{K}_1(2\sqrt{z})\ _1F_2(1;p\!+\!2,p\!+\!2;z)}{(p+1)^2},\tag{12}\label{Bessel1}
\end{align}
where in $(a)$ we substitute $q=2\sqrt{z}$, $(b)$ comes after some algebraic manipulations of the results given in \cite[Sec. 1.1.2]{Rosenheinrich.2012} along with the definition of $_xF_y$. Finally,  in $(c)$ we return to variable $z$, and substituting \eqref{Bessel1} into \eqref{AP1} we attain 
\eqref{AP2} given that $\lim\limits_{z\rightarrow 0}z^p\operatorname{K}_t(2\sqrt{z})=0,\ p>\tfrac{1}{2},\ t\in\{0,1\}$ (which comes from its series expansion at $z=0$), which concludes the proof.


\bibliographystyle{IEEEtran}
\bibliography{IEEEabrv,references}
\end{document}